\documentclass[]{paper}

\usepackage{afterpage,amsfonts,amsmath,amssymb,array}
\usepackage{epsfig,hhline,longtable}
\usepackage{tabularx}
\usepackage{algorithm,algorithmic}
\usepackage{amsthm}

\newtheorem{lemma}{Lemma}[section]

\newtheorem{proposition}{Proposition}[section]

\begin{document}
%\doi{10.1080/1350486X.20xx.CATSid}
%\issn{1466-4313}  \issnp{1350-486X}
%\jvol{00} \jnum{00} \jyear{2010} %\jmonth{January--March}

\markboth{Jan Baldeaux and Alexander Badran}{Applied Mathematical Finance}

%\articletype{GUIDE}

\title{Consistent Modelling of VIX and Equity Derivatives Using a $3/2$ plus Jumps Model}

\author{Jan Baldeaux and Alexander Badran}

%\affil{$^\ast$University of Technology, Sydney, Finance Discipline Group, PO Box 123, Broadway, NSW, 2007, Australia }

%\received{v1.1 released August 2010}

\maketitle

\begin{abstract}
The paper demonstrates that a pure-diffusion $3/2$ model is able to capture the observed upward-sloping implied volatility skew in VIX options.  This observation contradicts a common perception in the literature that jumps are required for the consistent modelling of equity and VIX derivatives. The pure-diffusion model, however, struggles to reproduce the smile in the implied volatilities of short-term index options. One remedy to this problem is to augment the model by introducing jumps in the index. The resulting $3/2$ plus jumps model turns out to be as tractable as its pure-diffusion counterpart when it comes to pricing equity, realized variance and VIX derivatives, but accurately captures the smile in implied volatilities of short-term index options.

\end{abstract}

\section{Introduction}

The Chicago Board Options Exchange Volatility Index (VIX) provides investors with a mechanism to gain direct exposure to the volatility of the S\&P500 index without the need for purchasing index options.  Consequently, the trading of VIX derivatives has become popular amongst investors.  In 2004 futures on the VIX began trading and were subsequently followed by options on the VIX in 2006.  Furthermore, since the inception of the VIX, volatility indices have been created to provide the same service on other indices, in particular, the VDAX and the VSTOXX, which are based on the DAX and the Euro STOXX 50 indices respectively.  Since derivative products are traded on both the underlying index and the volatility index, it is desirable to employ a model that can simultaneously reproduce the observed characteristics of products on both indices. Models that are capable of capturing these joint characteristics are known as consistent models. %One such requirement of our model is the ability to reproduce the observed implied-volatility smile of index options, while producing an upward-sloping skew for VIX options \cite{Sepp08b}.

A growing body of literature has been devoted to the joint modelling of equity and VIX derivatives.    The literature can generally be classed in terms of two approaches.   In the first approach, once the instantaneous dynamics of the underlying index are specified under a chosen pricing measure, the discounted price of a derivative can be expressed as a local martingale.  This is the approach adopted in \cite{LianZhu11}, \cite{Sepp08b}, \cite{ZhangZhu06} and \cite{ZhuLian11}.  \cite{ZhangZhu06} derived an analytic formula for VIX futures under the assumption that the S\&P500 is modelled by a Heston diffusion process \cite{Heston93}. A more general result was obtained in \cite{LianZhu11}.  Through a characteristic function approach these authors provided exact solutions (dependent upon a Fourier inversion) for the price of VIX derivatives when the S\&P500 is modelled by a Heston diffusion process with simultaneous jumps in the underlying index and the volatility process. A square-root stochastic variance model with variance jumps and time-dependent parameters was considered for the evolution of the S\&P500 index in \cite{Sepp08b}. The author provided formulae for the pricing and hedging of a variety of volatility derivatives.  Alternatively there is the ``market-model" approach, where variance swaps are modelled directly, as is done in \cite{Bergomi05} and \cite{ContKokholm}. The latter authors proposed a flexible market model that is capable of efficiently pricing realized-variance derivatives, index options, and VIX derivatives.  Realized-variance derivatives were priced using Fourier transforms, index derivatives were priced using a mixing formula, which averages Black-Scholes model prices, and VIX derivatives were priced, subject to an approximation, using Fourier-transform methods.

Models considered under the first approach generally yield (quasi-)closed-form solutions for derivative prices, which by definition are tractable. The challenge lies in ensuring that empirically observed facts from the market data, i.e. characteristic features of the joint dynamics of equity and VIX derivatives, are captured.  On the other hand, the market-model approach ensures by construction that models accurately reflect observed empirical characteristics.  The challenge remaining is to obtain an acceptable level of tractability when pricing derivative products. In this paper we follow the first approach and consider the joint modelling of equity and VIX derivatives when the underlying index follows a 3/2 process \cite{CarrSu07, Heston97, ItkinCa10, Lewis00} with jumps in the index only (henceforth called the 3/2 plus jumps model).   The model presented here is more parsimonious than competing models from its class; it is able to accurately capture the joint dynamics of equity and VIX derivatives, while retaining the advantage over market models of analytic tractability.  We point out that this model was used in the context of pricing target volatility fund derivatives in \cite{Meyer12}.

The selection of a $3/2$ model for the underlying index is motivated by several observations in recent literature. There is both empirical and theoretical evidence suggesting that the $3/2$ model is a suitable candidate for modelling instantaneous variance. \cite{BJY06} conducted an empirical study on the time-series properties of instantaneous variance by using S\&P100 implied volatilities as a proxy.  The authors found that a linear-drift was rejected in favour of a non-linear drift and estimated that a variance exponent of approximately 1.3 was required to fit the data.  In a separate study, \cite{CarrSu07} proposed a new framework for pricing variance swaps and were able to support the findings of \cite{BJY06} using a purely theoretical argument. Furthermore, the excellent results obtained by \cite{Drimus11}, who employed the 3/2 model to price realized-variance derivatives, naturally encourage the application of the 3/2 framework to VIX derivatives.  Despite having a qualitative advantage over other stochastic volatility models \cite{Drimus11}, the $3/2$ model, or any augmented version of this model, has yet to be applied to the consistent pricing of equity and VIX derivatives.  The final motivating factor is the claim that jumps must be included in the dynamics of the underlying index to capture the upward-sloping implied volatility skew of VIX options \cite{Sepp08b}.
%We point out that the augmented 3/2 model was employed in the context of pricing target volatility fund derivatives.

In related literature the only mention of the $3/2$ model in the context of VIX derivatives is in \cite{GoardMaz12}, where the problem is approached from the perspective of directly modelling the VIX. Closed-form solutions are found for VIX derivatives under the assumption that the VIX follows a $3/2$ process.  In this paper a markedly different approach is adopted.  Rather than specifying dynamics for the untradable VIX, without providing a connection to the underlying index, we follow the approach from \cite{LianZhu11}, where the dynamics of the underlying index are specified and an expression for the VIX is later derived. Our approach is superior to that of \cite{GoardMaz12}; issues of consistency are addressed directly and the model lends itself to a more intuitive interpretation.

The main contribution of this paper is the derivation of quasi-closed-form solutions for the pricing of VIX derivatives under the assumption that the underlying follows the 3/2 model.  The newly-found solutions retain the analytic tractability enjoyed by those found in the context of realized-variance products \cite{Drimus11}. The formulae derived in this paper allow for a numerical analysis to be performed to assess the appropriateness of the 3/2 framework for consistent modelling.  Upon performing the analysis we find that the pure-diffusion 3/2 model is capable of producing the commonly observed upward-sloping skew for VIX options.  This contradicts the previously made claims that pure-diffusion stochastic volatility models cannot consistently model VIX and equity derivatives \cite{Sepp08b}.  This desirable property distinguishes the 3/2 model from competing pure-diffusion stochastic volatility models.  We compare the 3/2 model to the Heston model and find that the latter produces downward-sloping implied volatilities for VIX options, whereas the former produces upward-sloping implied volatilities for VIX options.

Pure-diffusion volatility models, however, fail to capture features of implied volatility in equity options for short maturities \cite{Gatheral06}. To remedy this shortcoming jumps are introduced in the underlying index. The resulting $3/2$ plus jumps model is consequently studied in detail: first, by following the approach used for the pure-diffusion 3/2 model, we derive the conditions that ensure that the discounted stock price is a martingale under the pricing measure.  The novelty of this result is that we discuss whether a stochastic volatility model that allows for jumps is a martingale.  So far in the literature \cite{BayraktarKarXin11, Drimus11,Lewis00, MijatovicUrusov10} these results have been provided for pure-diffusion processes only, as they are based on Feller explosion tests \cite{KaratzasShreve00}.  Next, we produce the joint Fourier-Laplace transform of the logarithm of the index and the realized variance, which allows for the pricing of equity and realized-variance derivatives.  Though the $3/2$ model is not affine, we find that the joint Fourier-Laplace transform is exponentially affine in the logarithm of the stock price.  This allows for the simultaneous pricing of equity options across many strikes via the use of the Fourier-Cosine expansion method of \cite{FangOst08}. Such a finding is expected to significantly speed up the calibration procedure. The approach used in this paper is not restricted to the 3/2 plus jumps model and can be extended to a more general setting\footnote{The method is applicable to all conditionally Gaussian stochastic volatility models for which the Laplace transform of realized variance is known explicitly.}. In fact, we use this approach to obtain a closed-form solution for VIX options in the stochastic volatility plus jumps (SVJ) model, see  \cite{Bates96}, resulting in a small extension of the stochastic-volatility pricing formula presented in \cite{LianZhu11}.

The paper is structured as follows: in Section 2 we introduce the pure-diffusion $3/2$ model and present the empirical result that illustrates that this model is able to capture the joint characteristics of equity and index options.  We compare the pure-diffusion 3/2 model with the Heston model to highlight the difference in shape of the VIX implied volatilities.  The rest of the paper is concerned with the 3/2 plus jumps model.  Section 3 introduces the 3/2 plus jumps model and establishes the conditions that ensure that the discounted stock price is a martingale under the assumed pricing measure.  Next, characteristic functions for the logarithm of the index and the realized variance are derived. Finally, a quasi-analytic formula for call and put options on the VIX is derived.  Conclusions are stated in Section \ref{secconc}.

\section{Pure-Diffusion $3/2$ Model Applied to the VIX} \label{seccomp}

%So far, we have found the $3/2$ plus jumps model to be as tractable as the pure diffusion $3/2$ model. Hence the formulae in Section \ref{secVIXDer} were presented for teh more general case, the $3/2$ plus jumps model. The corresponding result for the $3/2$ model is obtained by simply setting $\lambda=0$ throughout Section \ref{secVIXDer}. On the other hand, for an empirical investigation, it is of course more natural to consider the more parsimonious version of the model, i.e. the $3/2$ model. We hence apply Proposition \ref{propderivVIX} setting $\lambda=04$, i.e. consider the pure diffusion $3/2$ model.
%
In this section we introduce the pure-diffusion $3/2$ model and present numerical results to illustrate that this model is able to produce upward-sloping implied volatility skews in VIX options.  On a probability space $( \Omega, \mathcal{F}, \mathbf{Q} )$, we introduce the risk-neutral dynamics for the stock price and the variance processes
%Under an assumed risk neutral pricing measure $\mathbf{Q}$, we assume the following dynamics for the stock price and the variance process,
\begin{align*}
 d S_t &= S_{t-} \left( r   dt +\rho \sqrt{V_t} dW^1_t + \sqrt{1 - \rho^2} \sqrt{V_t} dW^2_t\right)\, , \\
 d V_t &= \kappa V_t ( \theta - V_t) dt + \epsilon (V^{3/2}_t) d W^1_t \, ,
%\mbox{and}\,\,\,\, d V_t &= \kappa V_t ( \theta - V_t) dt + \epsilon (V^{3/2}_t) d W^1_t \, ,
\end{align*}
starting at $S_0 > 0 $ and $V_0 >0$ respectively, where $W = \left( W^1 \, , \, W^2 \right)$ is a two-dimensional Brownian motion under the risk-neutral measure. All stochastic processes are adapted to a filtration $\left( \mathcal{F}_t \right)_{t \in [0,T]}$ that satisfies the usual conditions with $\mathcal{F}_0$ being the trivial sigma field.  Furthermore, $r$ denotes the constant risk-free interest rate and $\rho$ the instantaneous correlation between the return on the index and the variance process.  As per usual, $\rho$ satisfies $-1 \leq \rho \leq 1$ and $\kappa$, $\theta$, and $\epsilon$ are assumed to be strictly positive. It is worth noting that unlike the Heston model the above model has a non-linear drift.  The speed of mean reversion is not constant, as is the case for the Heston model, but is now a stochastic quantity and is proportional to the instantaneous variance.

Throughout this paper we follow \cite{LianZhu11} and \cite{ZhangZhu06} and define the VIX via
\begin{equation} \label{eqdefVIX}
VIX^2_t := - \frac{2}{\tau} E \left( \ln \left( \frac{S_{t+\tau}}{S_t e^{r\tau}} \right) \vert \mathcal{F}_t \right) \times 100^2\, ,
\end{equation}
where $\tau= \frac{30}{365}$. We emphasize that this approach is superior to that of \cite{GoardMaz12}; instead of modelling the VIX directly, without providing a connection to the underlying index, the expression for the VIX in Equation \eqref{eqdefVIX} is derived directly from the dynamics of the underlying index.

Implied volatilities of VIX options exhibit a positive volatility skew, as stated in \cite{Sepp11}. The author asserts that ``SV [stochastic volatility models] without jumps are not consistent with the implied volatility skew observed in options on the VIX..." and that ``...only the {SV [stochastic volatility] model with appropriately chosen jumps} can fit the implied VIX skew". To assess these statements we calculate implied volatilities under the pure-diffusion $3/2$ model.  The price of a call option on the VIX is now given by
\begin{align}
\lefteqn{ e^{- r T} E \left( \left( VIX_T - K \right)^+ \right)}\nonumber
 \\ &= e^{- r T} \int^{\infty}_0 \left( \sqrt{ \frac{g(y, \tau)}{\tau} \times 100^2} - K \right)^+ \frac{1}{y^2} \frac{e^{\kappa \theta T}}{c(T)} p\left( \delta, \alpha, \frac{e^{\kappa \theta T}}{y c(T)}\right) \, dy,\label{eqpurediffcall}
\end{align}
where
 \begin{displaymath}
g (x, \tau) = - \frac{\partial}{\partial l} E \left( \exp \left( - l \int^{t + \tau}_t V_s ds  \right)\bigg|V_t =x  \right) \bigg|_{l=0} \,,
\end{displaymath}
$\delta= \frac{4(\kappa + \epsilon^2)}{\epsilon^2}$, $\alpha =\frac{1}{V_t c(T-t)}$, $c(t) = \epsilon^2 ( \exp \left( \kappa \theta t \right) - 1) / ( 4 \kappa \theta )$, and $p(\nu, \beta, \cdot)$ denotes the probability density function of a non-central chi-squared random variable with $\nu$ degrees of freedom, and non-centrality parameter $\beta$.  Equation \eqref{eqpurediffcall}
is a special case of the forthcoming Proposition \ref{propderivVIX}.

In order to provide the reader with parameters that are verifiable we use the parameters provided in \cite{Drimus11} for realized-variance derivatives.  Using Equation \eqref{eqpurediffcall} and the parameters
\begin{displaymath}
V_0 =0.2450^2 \, , \, \kappa = 22.84 \, , \, \theta = 0.4669^2 \, , \, \epsilon =8.56 \,,\mbox{ and } \, \rho=-0.99 \, ,
\end{displaymath}
we price VIX options for $T=3$ months and $T=6$ months. Then using Black's formula we find an implied volatility, $\zeta$, such that
\begin{displaymath}
E \left( \left( VIX_T - K \right)^+ \right) = E \left( VIX_T \right) N(d_1) - K N(d_2) \, ,
\end{displaymath}
where
\begin{displaymath}
d_1 = \frac{ \log ( E \left( VIX_T \right) / K) + \zeta^2 T }{\zeta \sqrt{T}} \, ,\mbox{ and }\, d_2 = d_1 - \sqrt{T} \zeta \, .
\end{displaymath}
\begin{figure}
\begin{center}
\includegraphics[scale=0.4,angle=0]{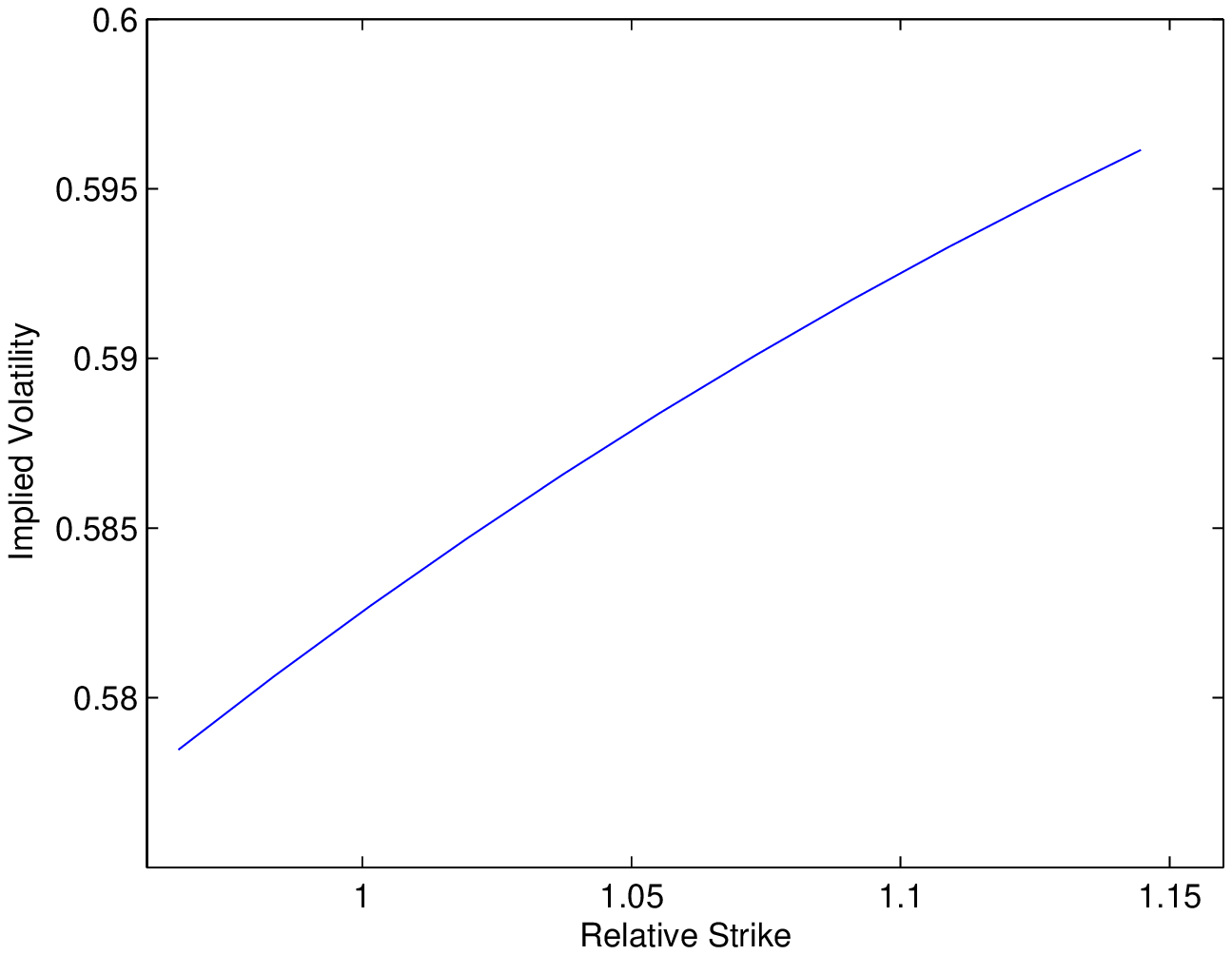}\includegraphics[scale=0.4,angle=0]{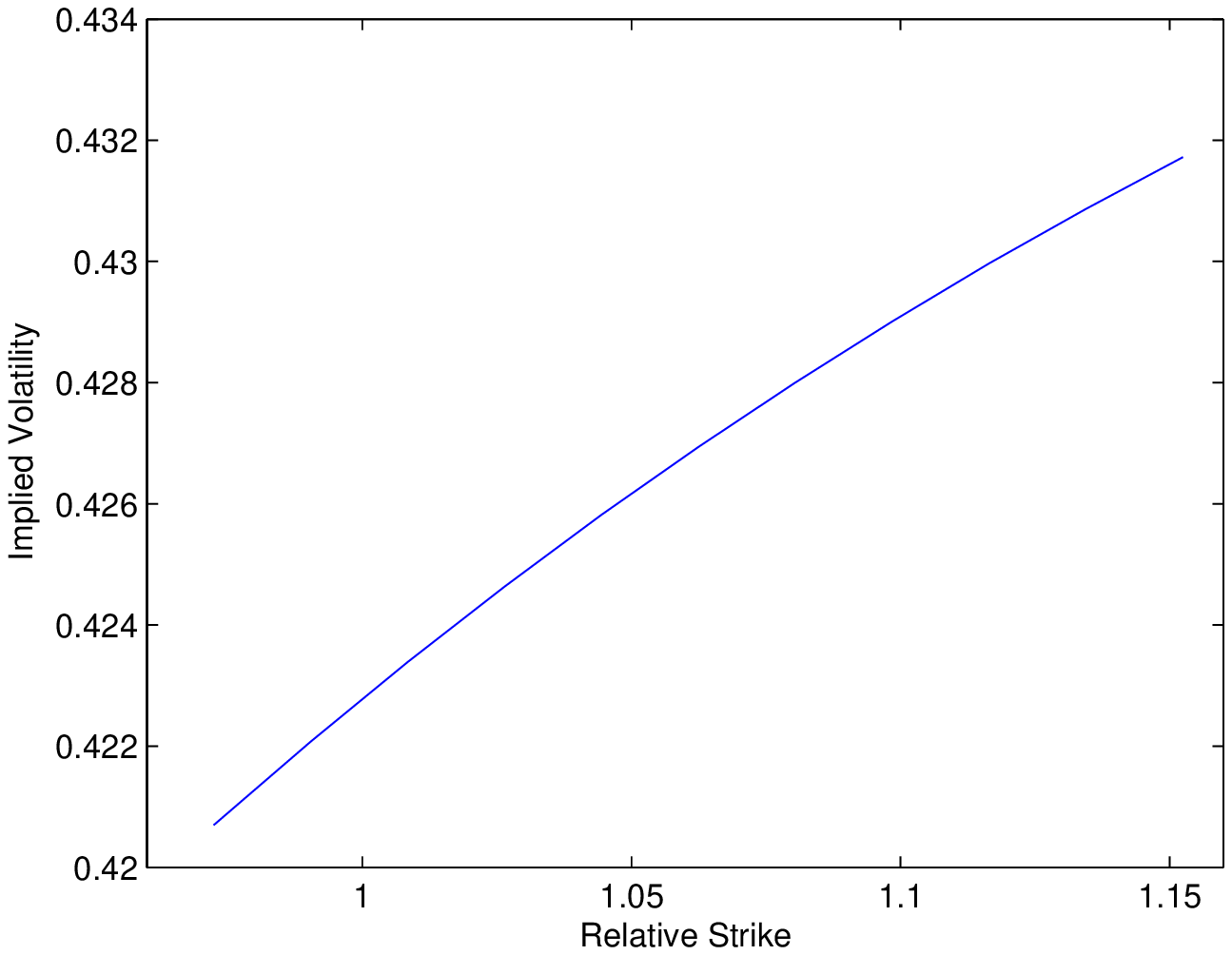}
\end{center}
\caption{\small Implied volatilities of call options on the VIX using the $3/2$ model parameters obtained in \cite{Drimus11} with $T=3$ months (left) and $T=6$ months (right).\label{fig3over2implvols}}
\end{figure}
The positive skew of the implied volatility of VIX options is shown in Figure \ref{fig3over2implvols} for maturities $T=3$ months and $T=6$ months, demonstrating that the dynamics of the pure-diffusion 3/2 model are in fact rich enough to fit the implied VIX skew. These observations support the findings of \cite{BJY06}, \cite{CarrSu07} and \cite{Drimus11} that suggest that the $3/2$ model is a good candidate for the pricing of volatility derivatives.

Next we compare the results produced by the $3/2$ model to the Heston model, which is commonly used for the pricing VIX derivatives \cite{LianZhu11, Sepp08b, ZhangZhu06, ZhuLian11}. A priori this seems to be a fair comparison; both models are stochastic volatility models, the two models have the same number of parameters, and enjoy the same level of analytical tractability.  To compute VIX option prices and the corresponding implied volatilities we use the pricing formula provided by \cite{LianZhu11} (see their Proposition 3).  Again, we use the following parameters obtained in \cite{Drimus11} for the Heston model,
\begin{displaymath}
V_0 = 0.2556^2 \, , \, \kappa = 3.8 \, , \, \theta=0.3095^2 \, , \, \epsilon =0.9288 \, , \mbox{ and } \rho=-0.7829 \, ,
\end{displaymath}
The result is shown in Figure \ref{figHestonimplvols}. Unlike for the $3/2$ model the implied volatilities are downward sloping, which is not consistent with market data.

\begin{figure}
\begin{center}
\includegraphics[scale=0.4,angle=0]{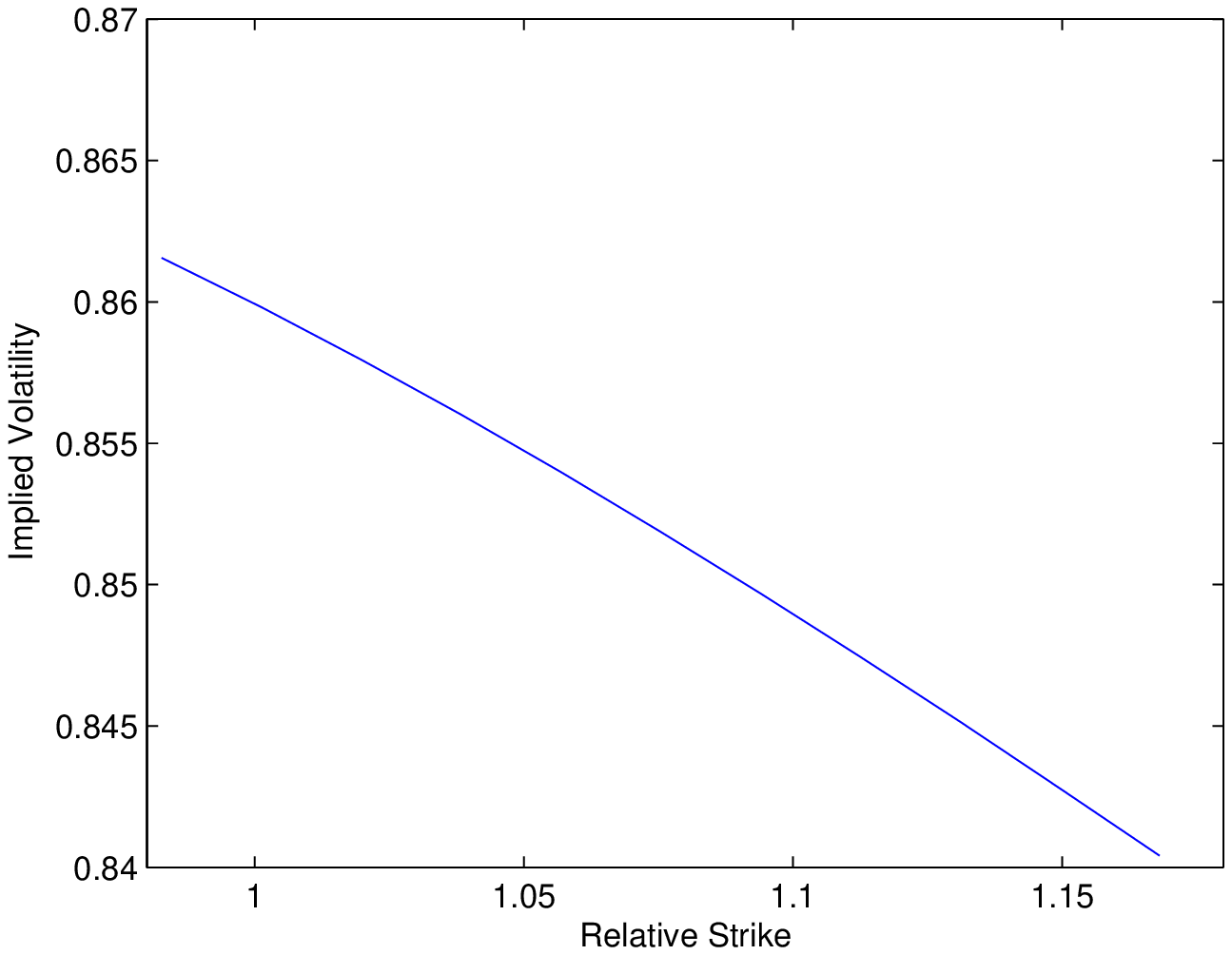}\includegraphics[scale=0.4,angle=0]{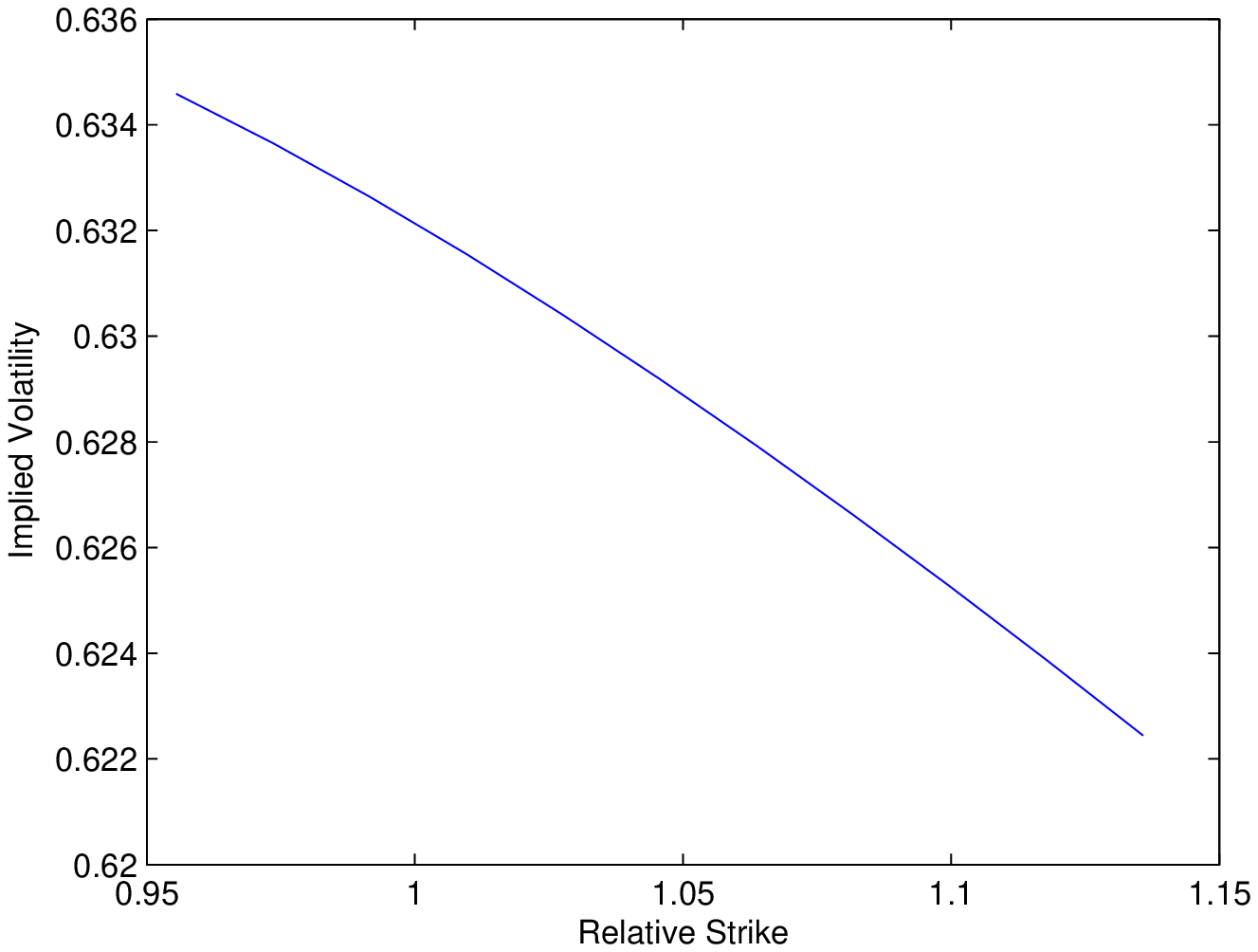}
\end{center}
\caption{\small Implied volatilities of call options on the VIX using the Heston model parameters obtained in \cite{Drimus11} with $T=3$ months (left) and $T=6$ months (right).\label{figHestonimplvols}}
\end{figure}

\section{The $3/2$ plus Jumps Model} \label{secintro}

The previous section demonstrated that the pure-diffusion $3/2$ model is capable of capturing the upward-sloping features of VIX option implied volatilities. Pure diffusion volatility models, however, fail to capture features of equity implied volatility for short expirations. To demonstrate this, we calibrate the pure-diffusion $3/2$ model to short-maturity S\&P500 option data. In Figure \ref{fig3over2fit} we present option implied volatilities for the S\&P500 on the 8th March 2012 for a maturity of nine days. The data set clearly exhibits a volatility smile. We find that though the pure-diffusion $3/2$ model captures the negative skew, it struggles to capture the smile.
\begin{figure}
\begin{center}
\includegraphics[scale=0.4,angle=0]{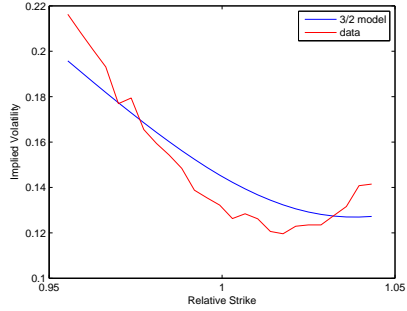}
\end{center}
\caption{\small Fit of the $3/2$ model to $9$ days S\&P500 implied volatilities on $8$ March, 2012. Model parameters obtained $\epsilon = 70.56$, $V_0 = 0.19^2$, $\kappa =30.84$, $\theta = 0.48^2$, $\rho= -0.55$. \label{fig3over2fit}}
\end{figure}

This observation motivates the extension of the model to allow for jumps in the underlying index in order to obtain better fit for short expirations. Consider the dynamics for the underlying index given by
\begin{align}
\label{eqstockprice1} d S_t &= S_{t-} \left( ( r - \lambda \bar{\mu} ) dt +\rho \sqrt{V_t} dW^1_t + \sqrt{1 - \rho^2} \sqrt{V_t} dW^2_t + (e^{\xi}-1) d N_t \right) \, ,\\
\label{eqvarprocess} d V_t &= \kappa V_t ( \theta - V_t) dt + \epsilon (V^{3/2}_t) d W^1_t \, ,
\end{align}
where we denote by $N$ a Poisson process at constant rate $\lambda$, by $e^{\xi}$ the relative jump size of the stock and $N$ is adapted to a filtration $\left( \mathcal{F}_t \right)_{t \in [0,T]}$.  The distribution of $\xi$ is assumed to be normal with mean $\mu$ and variance $\sigma^2$. The parameters $\mu$, $\bar{\mu}$, and $\sigma$ satisfy the following relationship
\begin{displaymath}
		\mu = \log( 1 + \bar{\mu} ) - \frac{1}{2} \sigma^2 \, .
\end{displaymath}
All other stochastic processes and parameters have been introduced in Section \ref{seccomp}.
Integrating Equation \eqref{eqstockprice1} yields
\begin{equation} \label{eqintegratedstock}
S_t = \tilde{S}_t \prod^{N_t}_{j=1} e^{\tilde{\xi}_j}
\end{equation}
where
\begin{equation} \label{eqcontpartstock}
\tilde{S}_t = S_0 \exp \left( ( r - \lambda \bar{\mu} ) t - \frac{1}{2} \int^t_0 V_s ds + \rho \int^t_0 \sqrt{V_s} dW^1_s + \sqrt{1 - \rho^2}  \int^t_0 \sqrt{V_s}dW^2_s \right) \, ,
\end{equation}
and we use $\tilde{\xi}_j$ to denote the logarithm of the relative jump size of the $j$th jump. Since the model
\eqref{eqstockprice1}~-~\eqref{eqvarprocess} is not affine, Equation \eqref{eqintegratedstock} gives us an important starting point for our analysis. In particular, one can now determine if the discounted stock price is a martingale under our assumed pricing measure.
\begin{proposition} \label{propmartingality} Let $S$ and $V$ be given by Equations \eqref{eqstockprice1} and 		
	\eqref{eqvarprocess} respectively. Then the discounted stock price $\bar{S}_t = \frac{S_t}{e^{rt}}$ is a martingale under $\mathbf{Q}$, if and only if
\begin{equation} \label{eqmartingalecond}
\kappa- \epsilon \rho \geq - \frac{\epsilon^2}{2} \, .
\end{equation}
\end{proposition}
\begin{proof}
%See Appendix \ref{appMartingale}.
We compute
\begin{align}
		E \left( \bar{S}_T \vert \mathcal{F}_t \right)
 &= \bar{S}_t E \left( \exp \left( - \frac{1}{2} \int^T_t V_s ds + \rho \int^T_t \sqrt{V_s} dW^1_s + \sqrt{1 - \rho^2} \int^T_t \sqrt{V_s} dW^2_s \right) \bigg| \mathcal{F}_t   \right)\nonumber
\\ & \hspace{0.5cm}\times E \left( \prod^{N_T}_{j=N_t + 1} e^{\tilde{\xi}_j} \right)e^{-\lambda \overline{\mu}(T-t)}\nonumber
\\ &= \bar{S}_t E \left( \exp \left( - \frac{1}{2} \int^T_t V_s ds + \rho \int^T_t \sqrt{V_s} dW^1_s + \sqrt{1- \rho^2} \int^T_t \sqrt{V_s} dW^2_s \right) \bigg| \mathcal{F}_t \right) \, \label{eqpuredifff}.
\end{align}
Equation \eqref{eqpuredifff} is clearly independent of the jump component of $S$. Hence $\bar{S}$ is a martingale under $\mathbf{Q}$ if and only if the corresponding discounted pure-diffusion model, $\tilde{S}_t \frac{e^{\lambda \bar{\mu} t}}{e^{rt}}$, is a martingale under $\mathbf{Q}$. Since this question was answered in \cite{Drimus11}, see his Equation (4), the desired result follows.
\end{proof}
%\begin{remark}
Starting with \cite{Sin98}, there has been a growing body of literature dealing with the question of whether the discounted stock price in a particular stochastic volatility model is a martingale or a strict local martingale under the pricing measure, e.g. \cite{AndersenPiterbarg07}, \cite{BayraktarKarXin11}, \cite{Lewis00}, and \cite{MijatovicUrusov10}. The specification of the model, in particular Equation \eqref{eqintegratedstock}, allows for the application of the above results, which were all formulated for pure diffusion processes. %\end{remark}
We remark that Condition \eqref{eqmartingalecond} is the same as the one presented in \cite{Drimus11}. Besides analyzing the martingale property of the model \eqref{eqstockprice1}-\eqref{eqvarprocess} we also compute functionals, which are required for the pricing of equity and VIX derivatives.
\subsection{Equity and Realized-Variance Derivatives} \label{secequityrealvarder}
In this section we derive formulae for the pricing of equity and realized-variance derivatives under the $3/2$ plus jumps model.  {We demonstrate that by adding jumps to the 3/2 model a better fit to the short-term smile can be obtained without incurring a loss in analytic tractability.}  Consider
\begin{displaymath}
X_t := \log(S_t) \, , \, t \geq 0 \, ,
\end{displaymath}
and define the realized variance as the quadratic variation of $X$, i.e.
\begin{displaymath}
RV_T := \int^T_0 V_s ds + \sum^{N_T}_{j=1} (\tilde{\xi}_j)^2 \, , \, t \geq 0 \, ,
\end{displaymath}
where $RV_T$ denotes realized variance and $T$ denotes the maturity of interest. We have the following result, which is the analogue of Proposition 2.2 in \cite{Drimus11}.

\begin{proposition} \label{propjointFLtransform} Let $u \in \mathbb{R}$ and $l \in \mathbb{R}^+$. In the $3/2$ plus jumps model, the joint Fourier-Laplace transform of $X_T$ and $(RV_T - RV_t)$ is given by
\begin{eqnarray*}
\lefteqn{E \left( \exp \left( i u X_T - l (RV_T - RV_t) \right) \vert X_t , V_t \right) }\\
&=&\exp \left( i u \left(X_t+(r-\lambda\overline{\mu})(T-t)\right) \right) \frac{\Gamma( \gamma - \alpha)}{\Gamma( \gamma)} \left( \frac{2}{\epsilon^2 y(t, V_t)}  \right)^{\alpha}\\
&&\times\,M\left(\alpha, \gamma, \frac{-2}{\epsilon^2 y(t,V_t)} \right) \exp \left( \lambda (T-t) ( a -1) \right)\, ,
\end{eqnarray*}
where
\begin{equation*}
y(t,V_t) = V_t \frac{\left( e^{\kappa \theta (T-t)} -1 \right) }{\kappa \theta}\, ,
\end{equation*}
\begin{equation*}
\alpha = - \left( \frac{1}{2} - \frac{p}{\epsilon^2} \right) + \sqrt{\left( \frac{1}{2} - \frac{p}{\epsilon^2} \right)^2 + 2 \frac{q}{\epsilon^2}}\, ,
\hspace{0.5cm}\gamma = 2 \left( \alpha +1 - \frac{p}{\epsilon^2} \right)\, ,
\hspace{0.5cm}p = - \kappa + i \epsilon \rho u\, ,
\end{equation*}
\begin{equation*}
q = l + \frac{i u}{2} + \frac{u^2}{2}\, \hspace{0.25cm}\mbox{and}\hspace{0.25cm}
a = \frac{\exp \left( - \frac{2 l \mu^2 - 2 i \mu u + \sigma^2 u^2}{2 + 4 l \sigma^2} \right) }{\sqrt{1 + 2 l \sigma^2}} \, ,
\end{equation*}and $M(a,b,c)$ denotes the confluent hypergeometric function.
\end{proposition}

\begin{proof} The proof is completed by noting that
%\begin{align*}
%\lefteqn{E \left( \exp \left( i u X_T - l (RV_T - RV_t) \right) \bigg| X_t , V_t \right) }
%\\ &= \exp \left( i u X_t \right) E \left( \exp \left( i u \log\left( \frac{\tilde{S}_T}{\tilde{S}_t} \right) -l\int_t^TV_sds   \right) \bigg| V_t \right)
%  E \left( \exp \left( i u \sum^{N_T}_{j= N_t +1} \tilde{\xi}_j - l \sum^{N_T}_{j=N_t+1} ( \tilde{\xi}_j )^2 \right) \right) \, .
%\end{align*}
\begin{eqnarray*}
\lefteqn{E \left( \exp \left( i u X_T - l (RV_T - RV_t) \right) \bigg| X_t , V_t \right)}
  \\ & = & \exp \left( i u X_t \right) E \left( \exp \left( i u \log\left( \frac{\tilde{S}_T}{\tilde{S}_t} \right) -l\int_t^TV_sds   \right) \bigg| V_t \right)
\\ &&  E \left( \exp \left( i u \sum^{N_T}_{j= N_t +1} \tilde{\xi}_j - l \sum^{N_T}_{j=N_t+1} ( \tilde{\xi}_j )^2 \right) \right) \, .
\end{eqnarray*}
The first conditional expectation was computed in \cite{CarrSu07} and \cite{Lewis00} and is given by
\begin{align*}
E \left( \exp \left( i u \log\left( \frac{\tilde{S}_T}{\tilde{S}_t} \right) -l\int_t^TV_sds   \right) \bigg| V_t \right)
		&=\exp \left( i u(r-\lambda\overline{\mu})(T-t) \right) \frac{\Gamma( \gamma - \alpha)}{\Gamma( \gamma)} \left( \frac{2}{\epsilon^2 y(t, V_t)}  \right)^{\alpha}\\
		&\hspace{0.5cm}\times M\left(\alpha, \gamma, \frac{-2}{\epsilon^2 y(t,V_t)} \right) \, .
\end{align*}
Furthermore, it can be seen that
\begin{displaymath}
E \left( \exp \left( i u \tilde{\xi}_j - l \tilde{\xi}^2_j \right) \right) = \frac{ \exp \left( - \frac{ 2 l \mu^2 -2 i \mu u + \sigma^2 u^2 }{2 + 4 l \sigma^2} \right) }{\sqrt{1+ 2 l \sigma^2}}
\end{displaymath}
and for $c>0$
\begin{displaymath}
E \left( c^{N_T - N_t} \right) = \exp \left( \lambda (T-t) ( c-1)  \right) \, .
\end{displaymath}
\end{proof}
Equity and realized-variance derivatives can now be priced using Proposition \ref{propjointFLtransform}.  For equity derivatives pricing requires the performance of a numerical Fourier inversion, such as those presented in \cite{CarrMad99} and \cite{Lewis00}.   Furthermore, since the characteristic function of $X_T$ is exponentially affine in $X_t$, we can apply the Fourier-Cosine expansion method as described in Section 3.3 in \cite{FangOst08}.  This allows for the simultaneous pricing of equity options across many strikes, which is expected to significantly speed up the calibration procedure.   For realized-variance derivatives one can employ a numerical Laplace inversion, see \cite{CGMY05}, or the more robust control-variate method developed in \cite{Drimus11}. We comment that implied volatility approximations for small log-forward moneyness and time to maturity for the $3/2$ plus jumps model can be obtained from \cite{MedvedevSca07}, as their Proposition 3 covers the $3/2$ plus jumps model.
\begin{figure}
\begin{center}
\includegraphics[scale=0.4,angle=0]{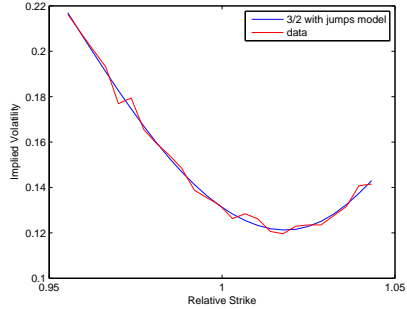}
\end{center}
\caption{\small Fit of the $3/2$ plus jumps model to $9$ days S\&P500 implied volatilities on $8$ March, 2012. Model parameters obtained $\epsilon = 50.56$, $V_0 = 0.0822^2$, $\kappa =30.84$, $\theta = 0.10^2$, $\rho= -0.57$, $\lambda=0.18$, $\mu=-0.30$, $\sigma=0.39$.\label{fig3over2plusjumpsfit}}
\end{figure}

%
%\begin{figure}
%\begin{center}
%\includegraphics[scale=0.4,angle=0]{threeovertwofit.eps}\includegraphics[scale=0.4,angle=0]{threeovertwoplusjumpsfit.eps}
%\end{center}
%\caption{\small Fit of the $3/2$ model and the $3/2$ plus jumps model to $9$ days S\&P500 implied volatilities on $8$ March, 2012. $3/2$ model parameters obtained $\epsilon = 70.56$, $V_0 = 0.19^2$, $\kappa =30.84$, $\theta = 0.48^2$, $\rho= -0.55$ (left). $3/2$ plus jumps model parameters obtained $\epsilon = 50.56$, $V_0 = 0.0822^2$, $\kappa =30.84$, $\theta = 0.10^2$, $\rho= -0.57$, $\lambda=0.18$, $\mu=-0.30$, $\sigma=0.39$ (right).\label{fig3over2plusjumpsfit}}
%\end{figure}
%
This section is concluded with a calibration of the $3/2$ plus jumps model to short-maturity S\&P500 option data. The inclusion of jumps improves the fit significantly as illustrated when comparing Figures \ref{fig3over2plusjumpsfit} and \ref{fig3over2fit}. The values for $V_0$ and $\theta$ decrease when we allow for jumps. Also, the parameters for the jump component are roughly in line with those obtained for SVJ models (see for example \cite{Gatheral06}).

\subsection{VIX Derivatives} \label{secVIXDer}

In this section we provide a general pricing formula for (European) call and put options on the VIX by extending the results of \cite{ZhangZhu06}.  The newly-found formula is then used for the pricing of VIX derivatives when the index follows a $3/2$ plus jumps process.  Of course, the results shown in Section \ref{seccomp} are obtained by setting the jump intensity $\lambda$ equal to $0$. We recall the definition of the VIX \eqref{eqdefVIX},
\begin{displaymath}
VIX^2_t := - \frac{2}{\tau} E \left( \ln \left( \frac{S_{t+\tau}}{S_t e^{r\tau}} \right) \vert \mathcal{F}_t \right) \times 100^2\, ,
\end{displaymath}
where $\tau= \frac{30}{365}$. The following result, which is an extension of Proposition 1 in \cite{ZhangZhu06}, allows for the derivation of a pricing formula for VIX options.

\begin{lemma} \label{propVIXformula} Let $S$, $V$, and $VIX^2$ be defined by Equations \eqref{eqstockprice1}, \eqref{eqvarprocess}, and \eqref{eqdefVIX}. Then
\begin{displaymath}
VIX^2_t = \left( \frac{g(V_t,  \tau)}{\tau}+2 \lambda ( \bar{\mu} - \mu) \right)  \times 100^2  \, , \, t \geq 0 \, ,
\end{displaymath}
where
\begin{displaymath}
g (x, \tau) = - \frac{\partial}{\partial l} E \left( \exp \left( - l \int^{t + \tau}_t V_s ds  \right)\bigg|V_t =x  \right) \bigg|_{l=0} \, .
\end{displaymath}
\end{lemma}
Lemma \ref{propVIXformula} is useful as it shows that the distribution of $VIX^2_t$ can be obtained via the distribution of $V_t$, for $t \geq 0$.  Consequently, the problem of pricing VIX derivatives is reduced to the problem of finding the transition density function for the variance process.  In the following proposition, we present the Zhang-Zhu formula for the futures price and a formula for call options.

\begin{proposition} \label{propderivVIX} Let $S$, $V$, and $VIX$ be given by Equations \eqref{eqstockprice1}, \eqref{eqvarprocess} and \eqref{eqdefVIX}. We obtain the following Zhang-Zhu formula for futures on the VIX
\begin{displaymath}
e^{- r T} E \left( VIX_T \right) = e^{- r T} \int^{\infty}_0 \sqrt{ \left( \frac{g(y, \tau)}{\tau}  +2 \lambda ( \bar{\mu} - \mu) \right) \times 100^2 } f_{V_T \vert V_0} (y) dy \, , \, T > 0 \, ,
\end{displaymath}
and the following formula for a call option
\begin{eqnarray*}
\lefteqn{ e^{- r T} E \left( \left( VIX_T - K \right)^+ \right)}
 \\ &=& e^{- r T} \int^{\infty}_0 \left( \sqrt{ \left( \frac{g(y, \tau)}{\tau}  +2 \lambda ( \bar{\mu} - \mu) \right) \times 100^2 } - K \right)^+ f_{V_T \vert V_0} (y) dy \, , \, T > 0 \, ,
\end{eqnarray*}
where $f_{V_T \vert V_0} ( y)$ denotes the transition density of $V$ started from $V_0$ at time $0$ being at $y$ at time $T$.
\end{proposition}
An expression for VIX put options can be obtained via the put-call parity relation for VIX options, namely
\begin{displaymath}
e^{- r T} E \left( \left( K - VIX_T \right)^+ \right) = e^{- r T} E \left( \left( VIX_T - K\right)^+ \right) + K e^{- r T} - e^{- r T} E \left( VIX_T \right) \, ,
\end{displaymath}
see Equation (25) in \cite{LianZhu11}.

For the $3/2$ model instead of using Lemma \ref{propVIXformula} we could alternatively use Theorem 4 in \cite{CarrSu07}.  However, our approach is not restricted to the $3/2$ model and applies to all stochastic volatility models for which the Laplace transform of the realized variance is known.
Furthermore, in the case of the $3/2$ model it is well known that $V_t$ is the inverse of a square-root processs \cite{B12, CarrSu07, Drimus11, GoardMaz12}.

\begin{lemma} \label{lemtransdensVt} Let $V$ be defined as in Equation \eqref{eqvarprocess}, then the transition density $f$ of $V$ is given by
\begin{displaymath}
f_{V_T \vert V_t} (y) = \frac{1}{y^2} \frac{e^{\kappa \theta (T-t)}}{c(T-t)} p\left( \delta, \alpha, \frac{e^{\kappa \theta (T-t)}}{y c(T-t)}\right) \, , \, T > t \geq 0,
\end{displaymath}
where $\delta= \frac{4(\kappa + \epsilon^2)}{\epsilon^2}$, $\alpha =\frac{1}{V_t c(T-t)}$, $c(t) = \epsilon^2 ( \exp \left( \kappa \theta t \right) - 1) / ( 4 \kappa \theta )$ and $p(\nu, \beta, \cdot)$ denotes the probability density function of a non-central chi-squared random variable with $\nu$ degrees of freedom, and non-centrality parameter $\beta$.
\end{lemma}
\begin{proof} As indicated above we introduce the process $X$ via $X_t = \frac{1}{V_t}$, whose dynamics are given by
\begin{displaymath}
d X_t = ( \kappa + \epsilon^2 - \kappa \theta X_t ) dt - \epsilon \sqrt{X_t} d W^1_t \, .
\end{displaymath}
Given $X(t)$, we note from \cite{JeanblancYoCh09} that
\begin{displaymath}
X_T \frac{ e^{\kappa \theta (T-t)}}{c(T-t)} \sim \chi^2 ( \delta, \alpha) \, , \, T > t \geq 0 \, ,
\end{displaymath}
where $\chi^2( \nu, \beta)$ denotes a non-central chi-squared random variable with $\nu$ degrees of freedom and non-centrality parameter $\beta$.
\end{proof}

Since we have an expression for the transition density of $V$, Proposition \ref{propderivVIX} can be used to price derivatives on the VIX as a discounted expectation.

To further demonstrate that the methodology presented in this section is not restricted to the $3/2$ plus jumps model, consider the stochastic volatility plus jumps model \cite{Bates96, DuffiePanSin00} given by
\begin{align}
d \tilde{S}_t &= \tilde{S}_{t-} \left( ( r - \lambda \bar{\mu} ) dt + \sqrt{\tilde{V}_t} \left( \rho dW^1_t + \sqrt{1 - \rho^2}d W^2_t \right)+ (e^{\xi}-1) d N_t  \right) \label{eqSVJstock}
\\ d \tilde{V}_t &= \tilde{\kappa} ( \tilde{\theta} - \tilde{V}_t) + \tilde{\epsilon} \sqrt{\tilde{V}_t} d W^1_t \, , \label{eqSVJvar}
\end{align}
where $r, \rho, \lambda, \bar{\mu}, \mu,$ and $\sigma$ are as defined for the $3/2$ model, and $\tilde{\kappa}, \tilde{\theta}$, and $\tilde{\epsilon} > 0$.
Using Lemma \ref{propVIXformula} we have
\begin{displaymath}
VIX^2_t =  \frac{g ( \tilde{V}_t, \tau )}{\tau} +2 \lambda ( \bar{\mu} - \mu) \, ,  %(a \tilde{V}_t + b )\times 100^2 \, ,
\end{displaymath}
where
\begin{displaymath}
g (x,\tau) = a x + b\, ,  \, a = \frac{1 - e^{- \tilde{\kappa} \tau}}{\tilde{\kappa}} \, , \, b = \tilde{\theta} ( \tau - a ) \, .
\end{displaymath}
%where
%\begin{align*}
%a &= \frac{1 - e^{- \tilde{\kappa} \tau}}{\tilde{\kappa}}  \, , %  \label{eqaSVJ}
%\\ b &= \tilde{\theta} ( \tau - a )  \, , %\label{eqbSVJ} \, ,
%\end{align*}
%
%\begin{align*}
%a &= \frac{1 - e^{- \tilde{\kappa} \tau}}{\tilde{\kappa} \tau} \, , %  \label{eqaSVJ}
%\\ b &= \tilde{\theta} ( 1 -a ) + \lambda c  \, , %\label{eqbSVJ} \, ,
%\\ c &= 2 \left( \bar{\mu} - \mu \right) \, , %\label{eqcSVJ}  \, ,
%\end{align*}
 As mentioned previously, it is well known that the transition density of a square-root process is non-central chi-squared. Therefore, Proposition \ref{propderivVIX} can be used to price options on the VIX in the setting \eqref{eqSVJstock} - \eqref{eqSVJvar}. This result is a small extension of Proposition 3 in \cite{LianZhu11}, as our result allows for jumps in the index.

%\section{Comparison of the $3/2$ and the Heston model} \label{seccomp}
\section{Conclusion} \label{secconc}

We derive general formulae for the pricing of equity and VIX derivatives.  The newly-found formulae allow for an empirical analysis to be performed to assess the appropriateness of the 3/2 framework for the consistent pricing of equity and VIX derivatives.  Empirically the pure-diffusion 3/2 model performs well; it is able to reproduce upward-sloping implied volatilities in VIX options, while a competing model of the same complexity and analytical tractability cannot.  Furthermore, the 3/2 plus jumps model is able to produce a better short-term fit to the implied volatility of index options than its pure-diffusion counterpart, without a loss in tractability.  These observations make the 3/2 plus jumps model a suitable candidate for the consistent modelling of equity and VIX derivatives.

%\setlinespacing{1.}

%\bibliographystyle{rAMF}
%\bibliography{references}

\end{document}